\documentclass[11pt, reqno]{amsart}
\usepackage{amsfonts,latexsym,enumerate}
\usepackage{amsmath}
\usepackage{amscd}
\usepackage{float,amsmath,amssymb,mathrsfs,bm,multirow,graphics}
\usepackage[dvips]{graphicx}
\usepackage[percent]{overpic}
\usepackage[pdftex]{color}
\usepackage{amsaddr}
\usepackage[numbers,sort&compress]{natbib}

\newcommand{\nequation}{\setcounter{equation}{0}}

\newcommand{\R}{{\Bbb R}}

\DeclareMathOperator{\Int}{int}

\def\XXint#1#2#3{{\setbox0=\hbox{$#1{#2#3}{\int}$}
\vcenter{\hbox{$#2#3$}}\kern-.5\wd0}}

\newtheorem{theorem}{Theorem}

\newtheorem{proposition}{Proposition}[section]
\newtheorem{corollary}[proposition]{Corollary}
\newtheorem{lemma}[proposition]{Lemma}
\newtheorem{definition}[proposition]{Definition}

\newtheorem{remark}[proposition]{Remark}
\newtheorem{example}[proposition]{Example}
\newtheorem{figuretext}{Figure}

\numberwithin{equation}{section}

\input epsf
\date{\today}
\title[Asymptotics to all orders of the Euler--Darboux equation in a triangle]
{Asymptotics to all orders of the Euler--Darboux equation in a triangle}

\author{Julian Mauersberger}
\address{Department of Mathematics, KTH Royal Institute of Technology, \\ 100 44 Stockholm, Sweden.}
\email{julianma@kth.se}

\begin{document}
\begin{abstract} 
\noindent
In Einstein's theory of relativity, the interaction of two collinearly polarized plane gravitational waves can be described by a Goursat problem for the Euler--Darboux equation in a triangular domain. In this paper, using a representation of the solution in terms of Abel integrals, we give a full asymptotic expansion of the solution near the diagonal of the triangle. The expansion is related to the formation of a curvature singularity of the spacetime. In particular, our framework allows for boundary data with derivatives which are singular at the corners. This level of generality is crucial for the application to gravitational waves.
\end{abstract}

\maketitle



\setcounter{tocdepth}{1}

\section{Introduction}\nequation
The collision of two collinearly polarized plane gravitational waves in Einstein's theory of relativity can be described mathematically \cite{G1991} by a Goursat problem for the Euler--Darboux equation \cite{M1973}
\begin{align}\label{linearernst}  
V_{xy} - \frac{V_x + V_y}{2(1-x-y)} = 0, \qquad (x,y) \in D,
\end{align}  
in the triangular region $D$ defined by (see Figure \ref{D.pdf})
\begin{align}\label{Ddef}
D = \{(x,y) \in \R^2\, | \, x \geq 0, \; y \geq 0, \; x+y < 1\}.
\end{align}

Besides its importance for collinearly polarized plane waves, equation \eqref{linearernst} can also be viewed as the linear limit of the nonlinear hyperbolic Ernst equation \cite{CF1984}. The hyperbolic Ernst equation is a reduction of the vacuum Einstein field equations which is similar to the elliptic version of the same equation \cite{E1968}. It describes the collision of two (not necessarily collinearly polarized) plane gravitational waves.  While many exact solutions of the hyperbolic Ernst equation have been found (see \cite{ET1989,EGA1988,FI1987b,NH1977}), it seems that the problem of determining the solution from given data has only been treated in the series of papers \cite{HE1989a, HE1989b, HE1990, HE1991} and in \cite{FST1999, AG2004,PM2017}. In \cite{HE1989a, HE1989b, HE1990, HE1991}, the authors relate the Goursat problem to the solution of a homogeneous Hilbert problem and in \cite{FST1999, AG2004, PM2017} the problem is studied by means of inverse scattering. In the case of collinear polarization, both \cite{FST1999} and \cite{HE1989a} provide representations for the solution in terms of Abel integrals. 

In the same way that the vacuum Einstein equations in the case of colliding plane gravitational waves can be reduced to the hyperbolic Ernst equation, the Einstein--Maxwell equations in vacuum can be reduced to an integrable system of two coupled nonlinear PDEs in the case of colliding electromagnetic plane waves \cite{E1968b, CX1985}. In the linear limit, both of these coupled equations reduce to the Euler--Darboux equation \eqref{linearernst}.

Of particular interest is the behavior of the solution of \eqref{linearernst} near the triangle's diagonal edge $x + y = 1$. Indeed, this behavior is related to the formation of a curvature singularity of the spacetime by the mutual focusing of the colliding waves. In \cite{S1972}, it was observed that the solution must behave like $\ln(1-x-y)$ as $x+y\to 1$. In this paper, using an Abel integral representation of the solution, we derive an  asymptotic expansion to all orders as $x+y\to 1$.

\begin{figure}
	\bigskip
	\begin{center}
		\begin{overpic}[width=.4\textwidth]{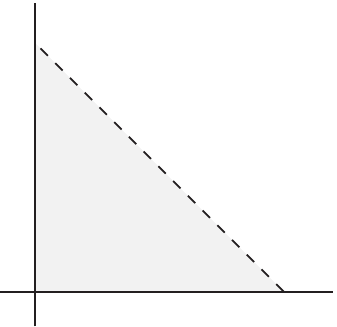}
			\put(29,33){ $D$}
			\put(103,9.7){\small $x$}
			\put(9.5,101){\small $y$}
			\put(82.5,3.5){\small $1$}
			\put(5,82){\small $1$}
			\put(27,3.5){\small $V(x,0) = V_0(x)$}
			\put(2,26){\small \rotatebox{90}{$V(0,y) = V_1(y)$}}
		\end{overpic}
		\begin{figuretext}\label{D.pdf}
			The triangular region $D$ defined in \eqref{Ddef} and the boundary conditions relevant for the Goursat problem. 
		\end{figuretext}
	\end{center}
\end{figure}
Our first result (Theorem \ref{linearmainth}) establishes a mathematically precise version of the classical Abel integral representation of the solution of the Goursat problem for \eqref{linearernst} with given boundary data. This type of representation is well known (cf. e.g. \cite{HE1989a, GS2002, FST1999}) and the main purpose of Theorem~\ref{linearmainth} is to provide a formulation that is suitable for our needs. We discuss regularity, the behavior at the boundary of $D$, and uniqueness of the solution under reasonable assumptions on the boundary data. In this context, reasonable means that our results can be applied to the common examples of collinear solutions (cf. e.g. \cite{FI1987,KP1971,S1972}). In particular, we allow for boundary data with singular derivatives at the corners of $D$. As a consequence, our representation of the solution of \eqref{linearernst} contains singular integrands, which is the main challenge in the analysis of the Goursat problem.

Our second result (Theorem \ref{GeneralThmAsymptotics}) describes the asymptotic behavior of the solution near the diagonal of $D$. We show that $V(x,1-x-\epsilon)$ admits an asymptotic expansion to all orders of the form
\begin{align*}
V(x,1-x-\epsilon) =\sum_{j=0}^{J}  f_j(x) \epsilon^j \ln(\epsilon) 
\; +\sum_{j=0}^{J} g_j(x)\epsilon^j   
+O(\epsilon^{J+1}\ln(\epsilon)), \qquad \epsilon \downarrow 0,
\end{align*}
where the coefficients $f_j,g_j$ are given explicitly in terms of the boundary data $V_0(x)=V(x,0)$ and $V_1(y)=V(0,y)$, and where the error term is uniform on compact subsets of~$(0,1)$. The first few terms of the asymptotic formula are given by
\begin{align}\label{formulanaiveasymptotics}\begin{split}
V(x,1-x-\epsilon) =& \; -\frac{1}{\pi} \left( h_0(x)+ h_1(x) \right) \ln(\epsilon)\\
&+\frac{1}{\pi} \frac{d}{dx} \left(\int_{0}^{x} h_0(k) \ln(4(x-k)) dk - \int_{x}^{1} \ln(4(k-x)) h_1(k)dk\right) \\
& -\frac{1}{2\pi} \left( h_0'(x)+ h_1'(x) \right) \epsilon \ln(\epsilon)\\ 
&+ \frac{1}{2\pi} \frac{d^2}{dx^2} \left(\int_{0}^{x} h_0(k) \ln(4(x-k)) dk - \int_{x}^{1} \ln(4(k-x)) h_1(k)dk\right) \epsilon \\
&+ O(\epsilon^2 \ln(\epsilon)), \qquad \epsilon \downarrow 0,
\end{split}
\end{align}
where
\begin{align}\label{Defh0h1}
h_0(k) = \sqrt{1-k} \int_0^k \frac{V_{0x}(x')}{\sqrt{k-x'}} dx', \hspace{0.5cm} 	h_1(k) = \sqrt{k} \int_0^{1-k} \frac{V_{1y}(y')}{\sqrt{1-k-y'}} dy'.
\end{align}

We finally point out that representations for solutions of boundary value problems similar to the Goursat problem for \eqref{linearernst} have been subject of recent mathematical research. For instance, such representations have been found for solutions of Protter problems for equations of Keldysh type in $(3+1)$ dimensions \cite{PHNS2017a,PHNS2017b}. In~\cite{PHNS2017a}, the representations have also been applied to derive asymptotic formulas for the solution. The Euler--Darboux equation and integral representations of its solution have also arisen in the context of hierarchies of integrable systems such as the KdV hierarchy~\cite{KMM2010}.

The two main results of the paper are presented in Section \ref{SectionMainresults}, and their proofs are given in Section \ref{SectionProof1} and Section \ref{SectionAsymptotics}, respectively.
In Section \ref{SectionApplication} we apply our results to collinearly polarized colliding gravitational waves and give full asymptotic expansions for the components of the Weyl tensor.

\section{Main results} \label{SectionMainresults}
Let $D$ be defined by \eqref{Ddef}. Since \eqref{linearernst} is a linear equation, we can assume that $V$ is real-valued and that $V(0,0)=0$. We introduce a notion of $C^n$-solution of the Goursat problem for \eqref{linearernst} in $D$ as follows. 

\begin{definition}\label{linearsolutiondef}\upshape
	Let $V_0(x)$, $x \in [0, 1)$, and $V_1(y)$, $y \in [0,1)$, be real-valued functions.
	A function $V:D \to \R$ is called a {\it $C^n$-solution of the Goursat problem for \eqref{linearernst} in $D$ with data $\{V_0, V_1\}$} if there exists an $\alpha \in [0,1)$ such that
	\begin{align*}
	\begin{cases}
	V \in C(D) \cap C^n(\Int(D)), 
	\\
	\text{$V(x,y)$ satisfies the Euler--Darboux equation \eqref{linearernst} in $\Int(D)$,}
	\\
	\text{$x^\alpha V_x, y^\alpha V_y, x^\alpha y^\alpha V_{xy} \in C(D)$},
	\\
	\text{$V(x,0) = V_0(x)$ for $x \in [0,1)$,}
	\\
	\text{$V(0,y) = V_1(y)$ for $y \in [0,1)$.}
	\end{cases}
	\end{align*}
\end{definition}

The following theorem solves the Goursat problem for \eqref{linearernst} in $D$ by providing a representation for the solution in terms of the boundary data. 

\begin{theorem}[Solution of the Euler--Darboux equation in a triangle]\label{linearmainth}
	Let $n \geq 2$ be an integer and suppose $\alpha \in [0,1)$. Let $V_0(x)$, $x \in [0, 1)$, and $V_1(y)$, $y \in [0,1)$, be two real-valued functions such that 
	\begin{align}\label{V0V1assumptions}
	\begin{cases}
	V_0, V_1 \in C([0,1)) \cap C^n((0,1)), 
	\\
	\text{$x^\alpha V_{0x}, y^\alpha V_{1y}, x^{\alpha+1}V_{0xx}, y^{\alpha+1} V_{1yy} \in C([0,1))$,} 
	\\
	V_0(0) = V_1(0) = 0.
	\end{cases}
	\end{align}
	Then
	\begin{align}\begin{split}\label{linearVrepresentation}
	V(x,y) = &\; \frac{1}{\pi} \int_0^x \frac{\sqrt{1-k}}{\sqrt{(1-y-k)(x-k)}} \bigg(\int_0^k \frac{V_{0x}(x')}{\sqrt{k - x'}}dx'\bigg) dk
	\\
	&  + \frac{1}{\pi} \int_{1-y}^1 \frac{\sqrt{k}}{\sqrt{(k - (1-y))(k-x)}} \bigg(\int_0^{1-k} \frac{V_{1y}(y')}{\sqrt{1 - y' - k}} dy'\bigg) dk, \qquad (x,y) \in D,
		\end{split}
	\end{align}  
	is a $C^{n-1}$-solution of the Goursat problem for \eqref{linearernst} in $D$ with data $\{V_0, V_1\}$ (and with the same $\alpha$). Furthermore, $V$ is the only solution in the sense of Definition \ref{linearsolutiondef}.
\end{theorem}

\begin{remark} \upshape
	In the case of $n=2$, Theorem \ref{linearmainth} states that $V$ is at least in $C^1$. However, the derivative $V_{xy}$ always exists and is continuous in $\Int D$. Furthermore, $x^\alpha y^\alpha V_{xy}$ is still in $C(D)$ in this case as it is required in Definition \ref{linearsolutiondef}. 
\end{remark}

\begin{remark}\upshape
	An alternative integral representation for the solution of the Goursat problem for the Euler--Darboux equation was already derived in \cite{S1972} using Riemann's classical method \cite{CH1962, G1964}. The representation \eqref{linearVrepresentation} relies on Abel integrals and is analogous to formulas derived in \cite{HE1989a}, whereas the formula of \cite{S1972} involves the Legendre function $P_{-1/2}$ of order $-1/2$. 
\end{remark}

The representation \eqref{linearVrepresentation} can be used to study the behavior of the solution $V(x,y)$ near the diagonal $x+y=1$, i.e. the behavior of $V(x,1-x-\epsilon)$ as $\epsilon \downarrow 0$. Letting $y=1-x-\epsilon$ in our representation formula \eqref{linearVrepresentation}, we find
\begin{align}\begin{split} \label{splitVasympt}
V(x,1-x-\epsilon) = &\; \frac{1}{\pi} \int_0^x \frac{\sqrt{1-k}}{\sqrt{(x-k+\epsilon)(x-k)}} \bigg(\int_0^k \frac{V_{0x}(x')}{\sqrt{k - x'}}dx'\bigg) dk
\\ 
&  + \frac{1}{\pi} \int_{x+\epsilon}^1 \frac{\sqrt{k}}{\sqrt{(k - x-\epsilon)(k-x)}} \bigg(\int_0^{1-k} \frac{V_{1y}(y')}{\sqrt{1 - y' - k}} dy'\bigg) dk \\ 
=:&\; X_1 + X_2.
\end{split}
\end{align} 

We define $h_0, h_1$ by \eqref{Defh0h1} and the kernels $K_0$ and $K_1$ by
\begin{align} \label{DefKernels}
K_0(u,t) = \frac{1}{\sqrt{u}\sqrt{u+t}}, \quad K_1(u,t)= \frac{1}{\sqrt{u}\sqrt{u-t}},
\end{align}
for $u>0$ and $-u<t<u$. Furthermore, we define for $j\ge 0$ the constants
\begin{align} \label{Defcj}
	c_j =  \frac{1}{2^j} \left( \prod_{l=0}^{j-1} (2l+1) \right) = (-1)^j \frac{\partial^j}{\partial t^j}K_0(1,0) =\frac{\partial^j}{\partial t^j}K_1(1,0) 
\end{align}
and
\begin{align} \label{DefH}
	H_0^j(x,k) = h_0(k) - \sum_{l=0}^{j} \frac{h_0^{(l)}(x) (-1)^l}{l!} (x-k)^l, \quad H_1^j(x,k) = h_1(k) - \sum_{l=0}^{j} \frac{h_1^{(l)}(x)}{l!} (k-x)^l.
\end{align}

The following theorem provides the asymptotics of $V(x,1-x-\epsilon)$ as $\epsilon \downarrow 0$ up to all orders.

\begin{theorem}[Asymptotic expansion to all orders] \label{GeneralThmAsymptotics}
	Let $J\ge 0$ be an integer. Suppose that $V_0,V_1\in C^{J+2}((0,1))$ satisfy the conditions \eqref{V0V1assumptions}. Then the unique solution $V(x,y)$ of the Goursat problem for \eqref{linearernst} with data $\{ V_0,V_1 \}$ enjoys the asymptotic expansion
	\begin{align*}
	V(x,1-x-\epsilon) =&-\frac{1}{\pi}\sum_{j=0}^{J}  \frac{c_j}{(j!)^2}(h_0^{(j)}(x)+h_1^{(j)}(x)) \epsilon^j \ln(\epsilon) \\
	&  +\frac{1}{\pi}\sum_{j=0}^{J} ((-1)^jA_j(x)+B_j(x)) \epsilon^j   \\
	&+O(\epsilon^{J+1}\ln(\epsilon)), \qquad \epsilon \downarrow 0, \quad 0<x<1,
	\end{align*}
	where the error term is uniform with respect to $x$ in compact subsets of $(0,1)$,
	\begin{align*}
	A_j(x) =& \; \frac{c_j}{j!}\left( \int_{0}^{x} \frac{H_0^j(x,k)}{(x-k)^{j+1}}dk +\sum_{l =0}^{j-1} \frac{(-1)^l}{l!(l-j)}h_0^{(l)}(x)x^{l-j} +\frac{(-1)^jh_0^{(j)}(x)}{j!} \ln(x) \right) \\
	&+ \frac{h_0^{(j)}(x)}{j!}\left( \int_0^1\frac{v^j}{\sqrt{v+1}\sqrt{v}}dv+\int_{1}^{\infty} v^j \left( \frac{1}{\sqrt{v+1}\sqrt{v}}-\sum_{l=0}^{j} \frac{1}{l!} (-1)^l c_l  v^{-l-1}\right)dv  \right)\\
	&+ \frac{h_0^{(j)}(x)}{j!}\sum_{l=0}^{j-1} (-1)^l \frac{c_l}{l! (l-j)}, 
\\
	B_j(x) =&\; \frac{c_j}{j!} \left( \int_{x}^{1} \frac{H_1^j(x,k)}{(k-x)^{j+1}}dk + \sum_{l=0}^{j-1}  \frac{1}{l!(l-j)} h_1^{(l)}(x) (1-x)^{l-j} + \frac{h_1^{(j)}(x)}{j!} \ln(1-x) \right)\\
	&+ \frac{h_1^{(j)}(x)}{j!}\left( \int_{1}^{\infty} v^j \left( \frac{1}{\sqrt{v-1}\sqrt{v}}-\sum_{l=0}^{j} \frac{c_l}{l!} v^{-l-1}\right)dv +\sum_{l=0}^{j-1} \frac{c_l}{l!(l-j)} \right),
	\end{align*}
	and $h_0$, $h_1$, $H_0^j$, $H_1^j$ are defined by \eqref{Defh0h1} and \eqref{DefH}. The first few terms can also be represented as in \eqref{formulanaiveasymptotics} (see Remark \ref{TheoremNaiveAsymptotics}).
\end{theorem}

\section{Proof of Theorem \ref{linearmainth}} \label{SectionProof1}
In order to prove Theorem \ref{linearmainth} we study some properties of Abel integrals.

\begin{definition} \upshape
	Let $\alpha \in [0,1)$ and suppose $k^{\alpha} h \in C([0,1))$. Then the \emph{Abel transform $\mathcal{A} h$ of the function $h$} is defined by
	\[
	\mathcal{A} h(x) = \int_{0}^{x} \frac{h(k)}{\sqrt{x-k}} dk, \qquad x \in [0,1).
	\]
\end{definition}

The following lemma shows that the Abel transform is well defined and gives some basic properties.

\begin{lemma} \label{PropertiesAbel}
	Let $n\ge 0$ and $\alpha \in [0,1)$. Suppose $h \in C^n((0,1))$ and $k^{\alpha+j} h^{(j)} \in C([0,1))$ for $0\le j \le n$. Then the Abel transform $f=\mathcal{A}h$ of $h$, defined by
	\[
	f(x)= \mathcal{A}h(x) = \int_{0}^{x} \frac{h(k)}{\sqrt{x-k}}dk,
	\]
	is in $C^n((0,1))$. Furthermore, it holds that $x^{\alpha-1/2+j} f^{(j)} \in C([0,1))$ for all $0\le j \le n$.
\end{lemma}

\begin{proof}
	First, we show that the function $f$ is well defined. Fix some $x\in (0,1)$. Then we can define the constant $M =\sup_{k\in [0,x]} |h(k)k^\alpha|$ and get
	\[
	|f(x)| \le \int_{0}^{x} \left| \frac{h(k)}{\sqrt{x-k}} \right| dk \le M \int_{0}^{x} \frac{1}{k^\alpha \sqrt{x-k}} dk.
	\]
	The transformation $u=k/x$ gives
	\[
	\int_{0}^{x} \frac{1}{k^\alpha \sqrt{x-k}} dk = \int_{0}^{1} \frac{x}{(ux)^\alpha \sqrt{x-ux}} du = x^{1/2-\alpha} \int_{0}^{1} \frac{1}{u^\alpha \sqrt{1-u}} du = x^{1/2-\alpha} B\left(1-\alpha,\frac{1}{2}\right),
	\]
	where $B$ denotes the Beta function. By the identity $B(x,y)=\Gamma(x)\Gamma(y)/\Gamma(x+y)$, $\Gamma$ denoting the Gamma function, this yields
	\begin{align}\label{KeyEstimate}
	\int_{0}^{x} \left| \frac{h(k)}{\sqrt{x-k}} \right| dk \le M \, x^{1/2-\alpha} \sqrt{\pi} \frac{\Gamma(1-\alpha)}{\Gamma(3/2-\alpha)}.
	\end{align}
	Hence $f$ is well defined for all $x\in(0,1)$.
	
	Next we show that $f$ is continuous on $(0,1)$. We apply the transformation $u=k/x$ and obtain
	\begin{align} \label{Substitution of Abel integral}
	f(x) = \int_{0}^{x} \frac{h(k)}{\sqrt{x-k}}dk = \sqrt{x}\int_{0}^{1} \frac{h(ux)}{\sqrt{1-u}} du.
	\end{align}
	Then it follows easily by Lebegue's dominated convergence theorem that $f$ is in $C((0,1))$.
	
	In order to show that $x^\gamma f \in C([0,1))$ for $\gamma = \alpha-1/2$ we show that the limit $\lim_{x\to0} x^\gamma f(x)$ exists. In a similar way as \eqref{KeyEstimate}, we obtain the estimate
	\[
	\inf_{k\in [0,x]}(k^\alpha h(k)) x^{1/2-\alpha} \sqrt{\pi} \frac{\Gamma(1-\alpha)}{\Gamma(3/2-\alpha)}
	\le f(x) \le \sup_{k\in [0,x]}(k^\alpha h(k)) x^{1/2-\alpha} \sqrt{\pi} \frac{\Gamma(1-\alpha)}{\Gamma(3/2-\alpha)}
	\]
	for all $x \in (0,1)$. Hence, by multiplying the inequality above with $x^\gamma$, we find
	\[
	\lim_{x\to0} x^{\alpha-1/2} f(x)= \lim_{k\to 0} (k^\alpha h(k))  \sqrt{\pi} \frac{\Gamma(1-\alpha)}{\Gamma(3/2-\alpha)}.
	\]
	In particular, the function $x^\gamma f$ can be continuously extended onto $[0,1)$.
	
	In order to show differentiability of $f$, we consider formula \eqref{Substitution of Abel integral} again. By applying Lebeugue's dominated convergence theorem to the difference quotient,  and, by using that $k^{1+\alpha}h' \in C([0,1))$, we obtain
	\begin{align} \begin{split}\label{Formula of derivative of Abel integral}
	\frac{d}{dx} \mathcal{A}h(x) &= \frac{1}{2\sqrt{x}} \int_{0}^{1} \frac{h(ux)}{\sqrt{1-u}} du + \sqrt{x} \int_{0}^{1} \frac{h'(ux) u}{\sqrt{1-u}} du \\
	&= \frac{1}{2x} \mathcal{A}h(x)+ \frac{1}{x} \mathcal{A}(kh')(x).
	\end{split}
	\end{align} 
	Since $kh' $ also fulfills the assumptions of the Lemma (but with reduced $n$), all remaining statements follow.
\end{proof}

\begin{remark} \label{RemarkHadamardFP}\upshape
	We will also use the following formula for the derivative of the Abel transform, given by
	\begin{align} \label{FormulaDerivative}
	\frac{d}{dx} \left( \int_{0}^{x} \frac{h(k)}{\sqrt{x-k}}dk\right) =\frac{h(\delta)}{\sqrt{x-\delta}}- \frac{1}{2} \int_{0}^{\delta} \frac{h(k)}{(x-k)^{3/2}}dk + \int_{\delta}^{x} \frac{h'(k)}{\sqrt{x-k}}dk
	\end{align}
	for an arbitrary $0<\delta<x$. This is obtained by splitting the integral, integrating by parts, and differentiating.
\end{remark}

\begin{remark} \label{RemarkUpperBoundirr}\upshape
	Note that the proof of Lemma \ref{PropertiesAbel} does not change if we replace the segment $[0,1)$ in the lemma by a segment $[0,b)$, where $0<b<1$.
\end{remark}

The following lemma will be useful in the proof of Theorem \ref{linearmainth}, since the Abel integrals in the integral representation \eqref{linearVrepresentation} depend on two variables.

\begin{lemma} \label{AbelWithParameter}
	Let $h \in C((0,1))$ and $x^{\alpha} h\in C([0,1))$ for some $\alpha<1$. Suppose that $g \in C(D)$ and that there exists a number $\beta \ge 0$ and a constant $C >0$ such that
	\[
	\left| g(y,k) \right| \le \frac{C}{(1-x-y)^{\beta}}, \hspace{0.5cm} \text{for all } (x,y) \in D \text{ and } k \in [0,x].
	\]  
	Then the function $f$ defined by
	\[
	f(x,y)=  \int_{0}^{x} g(y,k) \frac{h(k)}{\sqrt{x-k}}dk
	\]
	is in $C(\Int D)$. Furthermore, we have $x^{\alpha-1/2} f \in C(D)$.
\end{lemma}

\begin{proof}
	Let $U$ be an open set in $D\backslash \{ x=0 \}$ such that $\overline{U} \subset D\backslash \{ x=0 \}$, where $\overline{U}$ denotes the closure of $U$. Substituting $u=k/x$ and using the assumptions of the lemma yields
	\[
	\int_{0}^{x} \left| g(y,k) \frac{h(k)}{\sqrt{x-k}} \right| dk=  \sqrt{x}\int_{0}^{1} \left| g(y,ux) \frac{h(ux)}{\sqrt{1-u}} \right| du \le  \frac{C\sqrt{x}}{(1-x-y)^{\beta}} \int_{0}^{1} \frac{|h(ux)|}{\sqrt{1-u}} du.
	\]
	The  integrand $|h(ux)|/\sqrt{1-u}$ is bounded by an $L^1((0,1))$-function independent of $x$ and the term $(1-x-y)^{-\beta}$ is bounded by a constant on $U$. This shows that $f$ is in $C(D \backslash \{ x=0 \})$ by applying Lebegue's dominate convergence theorem. 
	
	In order to prove the last statement of the lemma, we observe
	\[
	\int_{0}^{1}\frac{x^\alpha|h(ux)|}{\sqrt{1-u}} du \le \sup_{k\in [0,x)} (k^\alpha |h(k)|) \int_0^1 \frac{1}{u^\alpha \sqrt{1-u}} du.
	\]
	Hence, if we fix a neighborhood $W$ around a point $(0,y) \in D$ such that $\overline{W}\subset D$, we can bound the expression
	\[
	x^{\alpha-1/2} \sqrt{x}\int_{0}^{1} \left| g(y,ux) \frac{h(ux)}{\sqrt{1-u}} \right| du
	\] 
	independently of $x$ and $y$. This shows that $x^{\alpha-1/2} f \in C(D)$.
\end{proof}

The next lemma gives a well-known inversion formula for Abel integrals. This will be needed in order to check the boundary conditions $V(x,0) = V_0(x)$ and $V(0,y)=V_1(y)$.

\begin{lemma} \label{InversionAbel}
	Let $h,f$ be as in Lemma \ref{PropertiesAbel}. Then we have
	\[
	h(k)= \frac{1}{\pi} \frac{d}{dk} \left( \int_{0}^{k}\frac{f(x)}{\sqrt{k-x}} dx \right), \qquad k\in (0,1).
	\]
\end{lemma}

\begin{proof}
	Due to Lemma \ref{PropertiesAbel} the integral 
	\[
	\int_{0}^{k}\frac{f(x)}{\sqrt{k-x}} dx
	\]
	exists and is differentiable on $(0,1)$. By the definition of $f$, we have
	\[
	\int_0^k \frac{f(x)}{\sqrt{k-x}} dx = \int_{0}^{k} \int_{0}^{x} \frac{h(\kappa)}{\sqrt{x-\kappa}\sqrt{k-x}} d\kappa dx = \int_{0}^{k} h(\kappa) \int_{\kappa}^{k} \frac{1}{\sqrt{x-\kappa}\sqrt{k-x}} dx d\kappa.
	\]
	We apply the transformation $y=(k-x)/(k-\kappa)$ and get 
	\begin{align*}
	\int_0^k \frac{f(x)}{\sqrt{k-x}} dx = \int_0^k h(\kappa) \int_{0}^{1} \frac{dy}{\sqrt{y}\sqrt{1-y}}d\kappa=\pi \int_{0}^{k}h(\kappa)d\kappa.
	\end{align*}
	This implies the desired statement.
\end{proof}

\begin{proof}[Proof of Theorem \ref{linearmainth}]
	It follows immediately by Lemma \ref{AbelWithParameter} with
	 \[
	 g(y,k) = \sqrt{1-k} /\sqrt{1-y-k}
	 \]
	 that $V$ is in $C(D)$ with $V(0,0) = 0$. Furthermore, letting
	\begin{align} \label{DefG}
	G(y,k) =  \frac{\sqrt{1-k}}{\sqrt{1-y-k}} \mathcal{A}V_{0x}(k),
	\end{align}
	we have
	\begin{align}\begin{split}\label{FirstStepDerivativeV}
	V_x(x,y) = &\; \frac{1}{\pi} \frac{\partial}{\partial x}\int_0^x \frac{G(y,k)}{\sqrt{x-k}} dk 
	\\ 
	&  + \frac{1}{\pi} \frac{\partial}{\partial x} \left( \int_{1-y}^1 \frac{\sqrt{k}}{\sqrt{(k - (1-y))(k-x)}} \bigg(\int_0^{1-k} \frac{V_{1y}(y')}{\sqrt{1 - y' - k}} dy'\bigg) dk \right).  
	\end{split}
	\end{align}
	Note that $x^{\alpha-1/2}\mathcal{A}V_{0x}(x)$ is in $C([0,1))$ by Lemma \ref{PropertiesAbel} and hence $G$ satisfies all assumptions in Lemma \ref{PropertiesAbel} in the sense of Remark \ref{RemarkUpperBoundirr}. Thus the first integral in \eqref{FirstStepDerivativeV} can be differentiated. For the second integral in \eqref{FirstStepDerivativeV}, it is clear that we can interchange the derivative and the integral sign. Under consideration of \eqref{Formula of derivative of Abel integral} we find
	\begin{align*}
	V_x(x,y) =& \frac{1}{2x\pi}\mathcal{A}\left( G(y,\cdot) \right)(x) + \frac{1}{x\pi} \mathcal{A}\left(  kG_k(y,\cdot) \right)(x)\\ &+\frac{1}{2\pi}\mathcal{A}\left(  \frac{\sqrt{1-(\cdot)}}{(1-x-(\cdot))^{3/2}} \mathcal{A}V_{1y}(\cdot) \right)(y).
	\end{align*}
	
	Since Lemma \ref{AbelWithParameter} can be applied to all three terms (with different choices of $g$, $h$ and in the third case with $x$ and $y$ interchanged), we obtain that $V_x \in C(\Int D)$ and $x^\alpha V_x \in C(D)$. The cases $V_y$ and $V_{xy}$ are similar and we find $V \in C^{n-1}(\Int D)$ by induction.
	
	In order to prove that $V$ satisfies equation \eqref{linearernst}, we use the alternative formula \eqref{FormulaDerivative} to represent the derivative.  We only prove that the first term
	\[
	X_1(x,y) =  \frac{1}{\pi} \int_0^x \frac{\sqrt{1-k}}{\sqrt{(1-y-k)(x-k)}} \bigg(\int_0^k \frac{V_{0x}(x')}{\sqrt{k - x'}}dx'\bigg) dk
	\]of \eqref{linearVrepresentation} satisfies the equation, since the other term is similar. We fix some $(x,y) \in \Int D$ and define $G(y,k)$ by \eqref{DefG}.
	By \eqref{FormulaDerivative} it holds for any fixed $\delta \in (0,x)$ that
	\begin{align*}
	\pi X_{1x}(x,y)=&\frac{G(y,\delta)}{\sqrt{x-\delta}}- \frac{1}{2} \int_{0}^{\delta} \frac{G(y,k)}{(x-k)^{3/2}}dk + \int_{\delta}^{x} \frac{G_k(y,k)}{\sqrt{x-k}}dk, \\
	\pi X_{1y}(x,y) =&\int_{0}^{x} \frac{G_y(y,k)}{\sqrt{x-k}}dk,  \\
	\pi X_{1xy}(x,y) =& \,\frac{G_y(y,\delta)}{\sqrt{x-\delta}}- \frac{1}{2} \int_{0}^{\delta} \frac{G_y(y,k)}{(x-k)^{3/2}}dk + \int_{\delta}^{x} \frac{G_{yk}(y,k)}{\sqrt{x-k}}dk.
	\end{align*}
	In particular,
	\begin{align*}
	\pi X_{1x}(x,y) &= \lim_{\epsilon\rightarrow 0} \left( \frac{G(y,x)}{\sqrt{\epsilon}} - \frac{1}{2} \int_{0}^{x-\epsilon} \frac{G(y,k)}{(x-k)^{3/2}}dk\right), \\
	\pi X_{1xy}(x,y) &=\lim_{\epsilon\rightarrow 0} \left( \frac{G_y(y,x)}{\sqrt{\epsilon}} - \frac{1}{2} \int_{0}^{x-\epsilon} \frac{G_y(y,k)}{(x-k)^{3/2}}dk \right).
	\end{align*}
	The identity $G_y(y,k) = (2(1-y-k))^{-1}G(y,k)$ yields
	\begin{align*}
	&\pi X_{1x}(x,y) +\pi X_{1y}(x,y)\\
	=&\lim_{\epsilon\rightarrow 0} \left( \frac{G(y,x)}{\sqrt{\epsilon}} - \frac{1}{2} \int_{0}^{x-\epsilon} \frac{G(y,k)(1-y-k)-G(y,k)(x-k)}{(1-y-k)(x-k)^{3/2}}dk\right)\\
	=&\; 2\pi (1-x-y) X_{1xy}(x,y).
	\end{align*}
	The other term of \eqref{linearVrepresentation} is similar and hence $V$ satisfies \eqref{linearernst}.
	
	Next we prove that $V$ satisfies the boundary conditions. We have
	\[
	V(x,0) = \frac{1}{\pi} \mathcal{A} (\mathcal{A}V_{0x})(x).
	\]
	Furthermore, Lemma \ref{InversionAbel} (with $h=V_{0x}$ and $f=\mathcal{A}V_{0x}$) gives
	\[
	V_{0x}(x)=\frac{1}{\pi}\frac{d}{dx} \mathcal{A} (\mathcal{A}V_{0x})(x)=V_x(x,0) 
	\]
	for all $x \in (0,1)$ and hence $V(x,0) = V_{0}(x)+c$ for a constant $c \in \R$. But $0 = V(0,0) = V_{0}(0)$. Thus we have $c=0$ and $V(x,0) = V_{0}(x)$. The case $V(0,y)$ is completely analogous.
	
	It remains to show that $V$ is the only solution. For that, it suffices to show that a function $V$ satisfying \eqref{linearernst} with $V(x,0) = 0 = V(0,y)$ must vanish in $D$. Let $D_T = \{ (x,y) \in D: x+y < T \}$ for $0<T<1$. Then the coefficients
	\[
	\frac{1}{2(1-x-y)}
	\]
	of \eqref{linearernst} are bounded on $D_T$. It follows by the method of successive approximations \cite[pp. 136--137]{G1964} that $V$ is unique on $D_T$ and hence vanishes on $D_T$. Since $T$ was arbitrary in $(0,1)$, the solution $V$ is zero on $D$. This completes the proof.
\end{proof}

\section{Proof of Theorem \ref{GeneralThmAsymptotics}} \label{SectionAsymptotics}

We define $h_0$, $h_1$, $K_0$, $K_1$, $H_0^j$, $H_1^j$, $X_1,X_2$, and $c_j$ by \eqref{Defh0h1}, \eqref{splitVasympt}, \eqref{DefKernels}, \eqref{Defcj}, and \eqref{DefH}. Note that, by Lemma \ref{PropertiesAbel}, $h_0$ and $h_1$ have the same order of regularity as $V_{0x}$ and $V_{0y}$, respectively, and that $k^{\alpha-1/2} h_0(k)$ and $k^{\alpha-1/2} h_1(k)$ are in $C([0,1))$. Furthermore, the kernels satisfy $K_i(au,at)=a^{-1} K_i(u,t)$ for $i=0,1$ and $a>0$.

\begin{proof}[Proof of Theorem \ref{GeneralThmAsymptotics}]
	Fix some compact subset $I \subset (0,1)$. Consider $X_1$ first. Then, for fixed $x \in I$, we have
	\[
	X_1 = \frac{1}{\pi}\int_{0}^{x} \frac{H_0^J(x,k) }{\sqrt{x-k}\sqrt{x-k+\epsilon}} dk  + \frac{1}{\pi} \sum_{j=0}^{J}\frac{h_0^{(j)}(x)(-1)^j}{j!}\int_{0}^{x}K_0(x-k,\epsilon)(x-k)^jdk.
	\]
	After the substitution $u=x-k$, the rightmost integral in the expression above can be written as
	\[
	\int_{0}^{x} u^jK_0(u,\epsilon) du = \int_{0}^{\epsilon} u^jK_0(u,\epsilon) du+\int_{\epsilon}^{x} u^jK_0(u,\epsilon) du = \epsilon^j \int_0^1v^jK_0(v,1)dv + \int_{\epsilon}^{x} u^jK_0(u,\epsilon) du,
	\]
	where we applied the transformation $u=\epsilon v$. Denoting
	\[
	R_{K_0,J}(u,\epsilon) = K_0(u,\epsilon)-\sum_{l=0}^{J} \frac{\epsilon^l}{l!} (-1)^l c_l u^{-l-1},
	\] 
	we find
	\begin{align} \begin{split} \label{EquationFirstPartAsym}
	\int_{\epsilon}^{x} u^jK_0(u,\epsilon) du =& \int_{\epsilon}^{x} u^j \left( K_0(u,\epsilon)-\sum_{l=0}^{J} \frac{\epsilon^l}{l!} (-1)^l c_l u^{-l-1}\right)du +\sum_{l=0}^{J} \frac{\epsilon^l}{l!} (-1)^l c_l\int_{\epsilon}^{x}u^{j-l-1}du \\
	=& \;  \epsilon^j \int_{1}^{\infty} v^j \left( K_0(v,1)-\sum_{l=0}^{J} \frac{1}{l!} (-1)^l c_l  v^{-l-1}\right)dv - \int_{x}^{\infty} u^j R_{K_0,J}(u,\epsilon)du\\  &+ \sum_{\substack{l=0\\ l \neq j}}^J \frac{\epsilon^l}{l!} (-1)^l c_l\frac{1}{j-l} (x^{j-l}-\epsilon^{j-l}) + \frac{\epsilon^j}{j!} (-1)^j c_j (\ln(x) - \ln(\epsilon)), \end{split}
	\end{align}
	where we again substituted $u=\epsilon v$. We can split the sum $\sum_{l=0}^{J} \frac{1}{l!} (-1)^l c_l  v^{-l-1}$ at $l=j$ and obtain
	\begin{align*}
	\int_{\epsilon}^{x} u^jK_0(u,\epsilon) du =& \; \epsilon^j \int_{1}^{\infty} v^j \left( K_0(v,1)-\sum_{l=0}^{j} \frac{1}{l!} (-1)^l c_l  v^{-l-1}\right)dv - \int_{x}^{\infty} u^j R_{K_0,J}(u,\epsilon)du\\  &+ \mathcal{C}_j(x)+\sum_{l=j+1}^{J} \frac{\epsilon^l}{l!} (-1)^l c_l\frac{1}{j-l} x^{j-l} - \sum_{l=0}^{j-1} \frac{\epsilon^l}{l!} (-1)^l c_l\frac{1}{j-l} \epsilon^{j-l}\\
	&+ \frac{\epsilon^j}{j!} (-1)^j c_j (\ln(x) - \ln(\epsilon)),
	\end{align*}
	where
	\[
	\mathcal{C}_j(x)= \sum_{l=0}^{j-1} \frac{\epsilon^l}{l!} (-1)^l c_l\frac{1}{j-l} x^{j-l}.
	\]
	In the first term of $X_1$ we can expand $K_0$ (since $H_0^J(x,k) \sim (x-k)^{J+1}$) and find
	\begin{align*}
	\int_{0}^{x} \frac{H_0^J(x,k) }{\sqrt{x-k}\sqrt{x-k+\epsilon}} dk &= \sum_{j=0}^{J} \frac{\epsilon^j}{j!} (-1)^j c_j  \int_{0}^{x} H_0^J(x,k) (x-k)^{-j-1} dk \\
	&+ \int_0^x H_0^J(x,k) R_{K_0,J}(x-k,\epsilon) dk.
	\end{align*}
	In the integrals $\int_0^x H_0^J(x,k)(x-k)^{-j-1}dk$, we write 
	\[
	H_0^J(x,k) = h_0(k) - \sum_{l=0}^{j}\frac{h_0^{(l)}(x)(-1)^l}{l!}(x-k)^l - \sum_{l=j+1}^{J}\frac{h_0^{(l)}(x)(-1)^l}{l!}(x-k)^l.
	\]
	Then the sums $\sum_{l = j+1}^J (\cdots)$ vanish together with the terms
	\[
	\sum_{j=0}^J \frac{h_0^{(j)}(x)(-1)^j}{j!}\mathcal{C}_j(x).
	\]
	It remains to show that the remainders are in $O(\epsilon^{J+1}\ln(\epsilon))$. We write
	\[
	R_{K_0,J}(u,\epsilon) = \frac{\partial^{J+1}K_0(u,\xi_\epsilon)}{\partial\epsilon^{J+1}} \frac{\epsilon^{J+1}}{(J+1)!} = \epsilon^{J+1}\frac{c_{J+1}}{(J+1)!} \frac{(-1)^{J+1}}{\sqrt{u}(\sqrt{u+\xi_\epsilon})^{2J+3}},
	\]
	for some $\xi_\epsilon \in [0,\epsilon]$. Then
	\[
	|R_{K_0,J}(u,\epsilon)| \le \epsilon^{J+1} \frac{C_J}{u^{J+2}}
	\]
	for some constant $C_J>0$. Hence
	\[
	\left| \int_{x}^{\infty} u^j R_{K_0,J}(u,\epsilon)du \right| \le \epsilon^{J+1} C_J \int_x^\infty u^{j-J-2} du = O(\epsilon^{J+1}), \qquad \epsilon \downarrow 0,
	\]
	uniformly on $I$, since $C_J$ does not depend on $x$. For the second remainder, we denote $m=\min I>0$ and $M=\max I <1$ and we split the integral
	\[
	\int_0^x H_0^J(x,k) R_{K_0,J}(x-k,\epsilon) dk = \int_{0}^{x/2} H_0^J(x,k) R_{K_0,J}(x-k,\epsilon) dk + \int_{x/2}^{x} H_0^J(x,k) R_{K_0,J}(x-k,\epsilon) dk.
	\]
	Since $H_0^J(x,k)k^\gamma$, where $\gamma = \alpha-1/2$, is bounded by a constant independently of $x$ on $[0,M/2]$, there exists a constant $C>0$ such that
	\begin{align*}
	\left|\int_{0}^{x/2} H_0^J(x,k) R_{K_0,J}(x-k,\epsilon) dk\right| &\le C \epsilon^{J+1} \int_{0}^{x/2} \frac{1}{k^{\gamma}(x-k)^{J+2}}dk  \\ &\le C \epsilon^{J+1} \left( \frac{2}{M} \right)^{J+2} \frac{M^{1-\gamma}}{1-\gamma} = O(\epsilon^{J+1}), \qquad \epsilon \downarrow 0,
	\end{align*}
	uniformly on $I$. For the other integral, we observe
	\[
	|H_0^J(x,k)| = \left|\int_{k}^{x} \frac{h_0^{(J+1)}(t)}{J!} (k-t)^{J} dt \right| \le \sup_{t \in [m/2,M]} |h_0^{(J+1)}(t)| \frac{(x-k)^{J+1}}{J!}
	\]
	for $k \in [x/2,x]$ and hence there exist constants $C,C'>0$, independent of $x$, such that
	\begin{align*}
	\left| \int_{x/2}^{x} H_0^J(x,k) R_{K_0,J}(x-k,\epsilon) dk \right| &\le C \int_{x/2}^{x} (x-k)^{J+1}\int_{0}^{\epsilon}\left| \frac{\partial^{J+1}}{\partial t^{J+1}}K_0(x-k,t) \right| (\epsilon-t)^{J} dt dk \\
	&= C'  \int_{x/2}^{x} \int_{0}^{\epsilon} \frac{(x-k)^{J+1}}{\sqrt{x-k}(\sqrt{x-k+t})^{2J+3}} (\epsilon-t)^{J} dt dk \\
	&\le C' \epsilon^{J+1} \int_{x/2}^{x}\int_{0}^{1} \frac{1}{(x-k)+\epsilon s}ds dk \\
	&=C' \epsilon^J \left( (x/2+\epsilon)\ln(x/2+\epsilon)-(x/2) \ln(x/2) - \epsilon \ln (\epsilon) \right).
	\end{align*}
	We observe
	\[
	\frac{ (x/2+\epsilon)\ln(x/2+\epsilon)-(x/2) \ln(x/2) }{\epsilon} = \ln(\xi_\epsilon) + 1
	\]
	for some $\xi_\epsilon \in [x/2,x/2+\epsilon]$ and hence the quotient is bounded on $I$. This implies that
	\[
	\int_{x/2}^{x} H_0^J(x,k) R_{K_0,J}(x-k,\epsilon) dk = \mathcal{O} (\epsilon^{J+1} \ln(\epsilon)), \qquad \epsilon \downarrow 0,
	\]
	uniformly on $I$. A similar procedure for $X_2$ completes the proof. 
\end{proof}

\begin{remark} \upshape \label{TheoremNaiveAsymptotics}
	The first few terms of the asymptotic expression found in Theorem \ref{GeneralThmAsymptotics} can also be written as in \eqref{formulanaiveasymptotics}, where $x\in (0,1)$ and where we assume that $V_0, V_1 \in C^3((0,1))$ satisfy the conditions \eqref{V0V1assumptions}. Indeed, consider the case of $X_1$; the case of $X_2$ is similar. By splitting
\[
\pi X_1= \int_{0}^{x} h_0(k)K_0(x-k,\epsilon)dk =  \int_{0}^{x/2} h_0(k)K_0(x-k,\epsilon)dk + \int_{x/2}^{x} h_0(k)K_0(x-k,\epsilon)dk,
\]
we can expand $K_0$ in the first term and in the second term we integrate by parts. This yields
\begin{align*}
\pi X_1= &\;\int_{0}^{x/2} \frac{h_0(k)}{x-k}dk +\left( - h_0(x)\ln(\epsilon) +2h_0(x/2)\ln(\sqrt{x/2}+\sqrt{x/2+\epsilon}) \right) \\
&+2\int_{x/2}^x \ln(\sqrt{x-k}+\sqrt{x-k+\epsilon})h_0'(k)dk + O(\epsilon),
\end{align*}
as $\epsilon \downarrow 0$. Since $h_0'$ is bounded on $[x/2,x]$, this gives
\begin{align*}
\pi X_1= &\int_{0}^{x/2} \frac{h_0(k)}{x-k}dk  - h_0(x)\ln(\epsilon) +h_0(x/2)\ln(2x)  \\
&+\int_{x/2}^x \ln(4(x-k))h_0'(k)dk + O(\ln(\epsilon)\epsilon), \qquad \epsilon \downarrow 0.
\end{align*}
In a similar way as \eqref{FormulaDerivative} we obtain the identity
\[
\frac{d}{dx} \left( \int_{0}^{x} h_0(k) \ln(4(x-k))dk \right) =\int_{0}^{x/2} \frac{h_0(k)}{x-k}dk+h_0(x/2)\ln(2x)+\int_{x/2}^x \ln(4(x-k))h_0'(k)dk.
\]
The following coefficients can be computed similarly by integrating by parts iteratively.
\end{remark}

Alternatively, \eqref{formulanaiveasymptotics} can be obtained directly from Theorem \ref{GeneralThmAsymptotics} by writing
\begin{align*}
\frac{d}{dx} \left( \int_{0}^{x} h_0(k) \ln(4(x-k))dk \right) &= \frac{d}{dx} \left( \int_{0}^{x} H^0_0(x,k) \ln(4(x-k))dk + \int_{0}^{x} h_0(x) \ln(4(x-k))dk \right) \\
&= \int_0^x \frac{H_0^0(x,k)}{x-k} dk - h_0'(x) \int_{0}^{x} \ln(4(x-k))dk \\
& \; \; \; + \frac{d}{dx} \left( h_0(x) \int_0^x \ln(4(x-k))dk \right) \\
&= \int_0^x \frac{H_0^0(x,k)}{x-k} dk + h_0(x) \ln(4) +h_0(x) \ln(x),
\end{align*}
which is equal to $A_0^0(x)$ in Theorem \ref{GeneralThmAsymptotics}. The cases of $h_1$ and the next coefficients are similar.

\begin{example}[Solution of Khan and Penrose] \upshape
	The solution of Khan and Penrose (cf.~\cite{G1991,KP1971}) is given by
	\begin{align*} 
	V(x,y) = -\ln \left(\frac{1+\sqrt{x}\sqrt{1-y} + \sqrt{y} \sqrt{1-x}}{1-\sqrt{x}\sqrt{1-y} -\sqrt{y} \sqrt{1-x}}  \right).
	\end{align*}	
	If we directly compute the asymptotic behavior of $V$ as $x+y\to1$ we find
	\begin{align}\label{SolutionKhanPenrose}
	V(x,1-x-\epsilon) = 2\ln(\epsilon) - \ln(16(1-x)x) - \frac{1-2x}{2(1-x)x} \epsilon +\frac{3(1-2x+2x^2)}{16(1-x)^2x^2} \epsilon^2 + O(\epsilon^3).
	\end{align}
	If we compute the asymptotics by applying Theorem \ref{GeneralThmAsymptotics}, we have $h_0=-\pi =h_1$ and 
	\begin{align*}
	A_0(x) &= -\pi (\ln(x) +\ln(4)), &
	A_1(x) &= \frac{\pi}{2x}, &
	A_2(x) &=\frac{3\pi}{16 x^2}, 
	\\
	B_0(x) &= -\pi(\ln(1-x) + \ln(4) )&
	B_1(x) &=\frac{\pi}{2(1-x)}, &
	B_2(x) &= \frac{3\pi}{16 (1-x)^2},
	\end{align*}
	which is consistent with \eqref{SolutionKhanPenrose}.
\end{example}

\section{Application to gravitational waves} \label{SectionApplication}
It is shown \cite{G1991} (Eq.~(10.2)) that the colliding gravitational wave problem for collinearly polarized plane waves reduces to solving the equation
\begin{align}
\label{GriffithsEq}V_{fg} + \frac{V_f + V_g}{2(f+g)} = 0
\end{align}
in the triangular region 
\[
\left\{(f,g) \in \R^2 \, | \, f \leq \frac{1}{2}, \; g \leq \frac{1}{2}, \; f + g > 0 \right\}.
\]
The change of variables $x = \frac{1}{2} - g$, $y = \frac{1}{2} - f$ transforms \eqref{GriffithsEq} into \eqref{linearernst}. The components of the Weyl tensor (cf. Eq. (10.4) in \cite{G1991}) are given by
\begin{align*}
\Psi_0^\circ & = -  \frac{g'(v)^2}{4}\bigg( 2V_{gg}+\frac{3}{f+g}V_g - (f+g)V_g^3 \bigg)
\\
& = -  \frac{g'(v)^2}{4}\bigg( 2V_{xx}-\frac{3}{1-x-y}V_x +(1-x-y)V_x^3 \bigg),
\\
\Psi_2^\circ & =  f'(u) g'(v)\bigg(V_fV_g-\frac{1}{(f+g)^2}\bigg)
\\
& =  f'(u) g'(v)\bigg( V_xV_y-\frac{1}{(1-x-y)^2} \bigg),
\\
\Psi_4^\circ & = - \frac{f'(u)^2}{4}\bigg(2V_{ff}+\frac{3}{f+g}V_f - (f+g)V_f^3\bigg)
\\
& = - \frac{f'(u)^2}{4}\bigg( 2V_{yy}-\frac{3}{1-x-y}V_y +(1-x-y)V_y^3 \bigg),
\end{align*}
where $u,v \ge 0$ are suitable null coordinates in the interaction region \cite[Chapter 6]{G1991}. Using that $f(u)$ and $g(v)$ have the form (see Eq. (7.6) in \cite{G1991}) 
\begin{align*}
f(u) & = \frac{1}{2} - (c_1u)^{n_1}, \qquad g(v) = \frac{1}{2} - (c_2v)^{n_2}, 
\\
f'(u) & = - c_1 n_1 (c_1u)^{n_1-1}
= - \frac{n_1}{u}(\frac{1}{2} - f)
= - \frac{n_1}{u}y=- c_1 n_1y^{1- \frac{1}{n_1}},
\\
 g'(v) &= - c_2 n_2 (c_2v)^{n_2-1}
= - c_2 n_2x^{1- \frac{1}{n_2}},
\end{align*}
with some constants $c_1,c_2,n_1,n_2$, we find
\begin{align} \begin{split} \label{ComponentsWeyl}
\Psi_0^\circ & = \frac{(c_2 n_2x^{1- \frac{1}{n_2}})^2}{4}\bigg( 2V_{xx}-\frac{3}{1-x-y}V_x +(1-x-y)V_x^3 \bigg),
\\
\Psi_2^\circ & = c_1 n_1y^{1- \frac{1}{n_1}}c_2 n_2x^{1- \frac{1}{n_2}}\bigg( V_xV_y-\frac{1}{(1-x-y)^2} \bigg),
\\
\Psi_4^\circ & =\frac{(c_1 n_1y^{1- \frac{1}{n_1}})^2}{4}\bigg( 2V_{yy}-\frac{3}{1-x-y}V_y +(1-x-y)V_y^3 \bigg).
\end{split}
\end{align}
Theorem \ref{GeneralThmAsymptotics} implies
\[
V(x,y) =  \sum_{j=0}^{N}f_j(x)\epsilon^j \ln (\epsilon) +\sum_{j=0}^{N} g_j(x) \epsilon^j + O\big(\epsilon^{N+1}\ln(\epsilon)\big)
\]
as $\epsilon=\epsilon(x,y)=1-x-y \to 0$ uniformly for $x,y$ in some neighborhood of the diagonal away from the corners, where 
\[
f_j(x) = \frac{-c_j}{\pi (j!)^2}\left(h_0^{(j)}(x)+h_1^{(j)}(x)\right), \quad g_j(x)=\frac{1}{\pi}\left((-1)^jA_j(x)+B_j(x)\right).
\]
 Due to symmetry of the representation \eqref{linearVrepresentation}, it holds that
\[
V(x,y) =  \sum_{j=0}^{N}f_j(1-y)\epsilon^j \ln (\epsilon) +\sum_{j=0}^{N} \tilde{g}_j(1-y) \epsilon^j + O\big(\epsilon^{N+1}\ln(\epsilon)\big)
\]
as $\epsilon=1-x-y \to 0$, where $\tilde{g}_j$ is defined by interchanging $h_0$ and $h_1$ in the definition of $g_j$. It is easy to see from the proof of Theorem \ref{GeneralThmAsymptotics} that the $y$-derivative of the remainder is $O((1-x-y)^N \ln(1-x-y))$ uniformly. Letting $\epsilon=\epsilon(x,y)=1-x-y$, this gives
\begin{align*}
V_y(x,y) &= -\sum_{j=0}^{N}j f_j(x)\epsilon^{j-1} \ln (\epsilon) - \sum_{j=0}^{N} \gamma_j (x) \epsilon^{j-1}
+ O\big((\epsilon^N \ln(\epsilon)\big), \quad \epsilon \downarrow 0,
\end{align*}
and
\begin{align*}
V_{yy}(x,y) &= \sum_{j=0}^{N}(j-1)jf_j(x)\epsilon^{j-2} \ln (\epsilon)+ \sum_{j=0}^{N} G_j(x)\epsilon^{j-2} + O\big(\epsilon^{N-1}\ln(\epsilon)\big), \quad \epsilon \downarrow 0,
\end{align*}
where $\gamma_j = jg_j +f_j$ and $G_j =(j-1)\gamma_j +jf_j$. In the same way we get
\begin{align*}
V_x(x,y) &= -\sum_{j=0}^{N}jf_j(1-y)\epsilon^{j-1} \ln (\epsilon) - \sum_{j=0}^{N} \tilde{\gamma}_j(1-y)\epsilon^{j-1}+ O\big(\epsilon^{N}\ln(\epsilon)\big), \quad \epsilon \downarrow 0,
\\
V_{xx}(x,y) &= \sum_{j=0}^{N}(j-1)jf_j(1-y)\epsilon^{j-2} \ln (\epsilon)  + \sum_{j=0}^{N} \tilde{G}_j(1-y) \epsilon^{j-2} + O\big(\epsilon^{N-1}\ln(\epsilon)\big), \quad \epsilon \downarrow 0,
\end{align*}
where $\tilde{\gamma}_j = j\tilde{g_j} +f_j$ and $\tilde{G}_j =(j-1)\tilde{\gamma}_j +jf_j$. Together with \eqref{ComponentsWeyl} this leads to full expansions of the components of the Weyl tensor near the diagonal. We have shown the following corollary of Theorem \ref{GeneralThmAsymptotics}.
\begin{corollary}[Asymptotics of the Weyl tensor] Let $N \ge 1$ be and integer. Suppose  $V_0,V_1 \in C^{N+2}((0,1))$ satisfy \eqref{V0V1assumptions}. Then the components of the Weyl tensor associated to the solution $V(x,y)$ of the Goursat problem for \eqref{linearernst} with data $\{ V_0,V_1 \}$ have the following asymptotic behaviour as $\epsilon=\epsilon(x,y)=1-x-y \to 0:$
	\begin{align*}
	\Psi_0^o(x,y)&=\frac{(c_2 n_2x^{1- \frac{1}{n_2}})^2}{4}\Bigg[ \sum_{j=0}^{N} (2\tilde{G}_j(1-y) +3\tilde{\gamma}_j(1-y))\epsilon^{j-2}\\ &\quad+ 	 \sum_{j=0}^{N}(j+2)jf_j(1-y)\epsilon^{j-2} \ln (\epsilon) 
	\\
	&\quad-\epsilon\left( \sum_{j=0}^{N}jf_j(1-y)\epsilon^{j-1} \ln (\epsilon) + \sum_{j=0}^{N} \tilde{\gamma}_j(1-y)\epsilon^{j-1} \right)^3+ O\big(\epsilon^{N-1}\ln(\epsilon)^3\big) \Bigg],
	\end{align*}
	\begin{align*}
	\Psi_2^o(x,y)&=c_1 n_1c_2 n_2y^{1- \frac{1}{n_1}}x^{1- \frac{1}{n_2}}\Bigg[ \sum_{j=0}^{N} \left( \sum_{k+l=j} \gamma_k(x) \tilde{\gamma}_l(1-y) \right)\epsilon^{j-2}  - \frac{1}{\epsilon^2} 
	\\
	&\quad + \sum_{j=0}^{N}\left( \sum_{k+l=j} kf_k(x) \tilde{\gamma}_l(1-y) + kf_k(1-y) {\gamma}_l(x)  \right)\epsilon^{j-2} \ln (\epsilon) 
	\\
	&\quad + \sum_{j=0}^{N} \left( \sum_{k+l=j} klf_k(x) {f}_l(1-y) \right)\epsilon^{j-2} \ln(\epsilon)^2
+ O\big(\epsilon^{N-1}\ln(\epsilon)^2\big) \Bigg],
\\
	\Psi_4^o(x,y)&=\frac{(c_1 n_1y^{1- \frac{1}{n_1}})^2}{4}\Bigg[ \sum_{j=0}^{N} (2{G}_j(x) +3{\gamma}_j(x))\epsilon^{j-2} 
	\\
	&\quad + \sum_{j=0}^{N}(j+2)jf_j(x)\epsilon^{j-2} \ln (\epsilon) 
	\\
	&\quad-\epsilon\left( \sum_{j=0}^{N}jf_j(x)\epsilon^{j-1} \ln (\epsilon) + \sum_{j=0}^{N} {\gamma}_j(x)\epsilon^{j-1} \right)^3 + O\big(\epsilon^{N-1}\ln(\epsilon)^3\big) \Bigg],
	\end{align*}
where the error terms are uniform for $(x,y) \in D$ bounded away from the corners of $D$.
\end{corollary}

\bigskip
\noindent
{\bf Acknowledgement} {\it The author thanks Jonatan Lenells for bringing this problem to his attention. Support is acknowledged from the European Research Council, Grant Agreement No. 682537.}

\bibliographystyle{plain}

\begin{thebibliography}{99}
\small

\bibitem{AG2004}
G. A. Alekseev and J. B. Griffiths, Collision of plane gravitational and electromagnetic waves in a Minkowski background: solution of the characteristic initial value problem,
{\it Class. Quantum Grav.} {\bf 21} (2004), 5623--5654. 

\bibitem{CF1984}
S. Chandrasekhar and V. Ferrari, On the Nutku--Halil solution for colliding impulsive gravitational waves, {\it Proc. Roy. Soc. Lond. A} {\bf 396} (1984), 55--74.

\bibitem{CX1985}
S. Chandrasekhar and B. C. Xanthopoulus, On colliding waves in the Einstein--Maxwell theory, {\it Proc. Roy. Soc. Lond. A} {\bf 398} (1985), 223--259.

\bibitem{CH1962}
 R. Courant and D. Hilbert, {\it Methods of mathematical physics. Vol. II: Partial differential equations},   Interscience Publishers, New York-London, 1962.


\bibitem{ET1989}
A. Economou and D. Tsoubelis, Multiple-soliton solutions of Einstein's equations, {\it J. Math. Phys.} {\bf 30} (1989), 1562--1569. 

\bibitem{E1968}
F. J. Ernst, New formulation of the axially symmetric gravitational field problem, {\it Phys. Rev.} {\bf 167} (1968), 1175--1178.

\bibitem{E1968b}
F. J. Ernst, New formulation of the axially symmetric gravitational field problem. II, {\it Phys. Rev.} {\bf 168} (1968), 1415--1417.

\bibitem{EGA1988}
F. J. Ernst, A. Garc{\'i}a D., and I. Hauser, Colliding gravitational plane waves with noncollinear polarization. III. 
{\it J. Math. Phys.} {\bf 29} (1988), 681--689. 

\bibitem{FI1987}
V. Ferrari and J. Iba\~{n}ez,  A new exact solution for colliding gravitational plane waves, {\it Gen. Rel. Grav.} {\bf 19} (1987), 383--404.

\bibitem{FI1987b}
V. Ferrari and J. Iba{\~n}ez, On the collision of gravitational plane waves: a class of soliton solutions, {\it Gen. Rel. Grav.} {\bf 19} (1987), 405--425. 

\bibitem{FST1999}
A. S. Fokas,  L.-Y. Sung, and D. Tsoubelis,  The inverse spectral method for colliding gravitational waves, 
 {\it Math. Phys. Anal. Geom.} {\bf 1} (1999), 313--330.


\bibitem{G1964}
\'{E}. Goursat, {\it A course in mathematical analysis. Vol. III, Part I: Variation of solutions. Partial differential equations of the second order}, Dover Publications, Inc., New York, 1964.

\bibitem{G1991}
J. B. Griffiths, {\it Colliding plane waves in general relativity}, Oxford Mathematical Monographs. Oxford Science Publications. The Clarendon Press, Oxford University Press, New York, 1991.


\bibitem{GS2002}
J. B. Griffiths and M. Santano-Roco, The characteristic initial value problem for colliding plane waves: the linear case, {\it Class. Quantum Grav.} {\bf 19} (2002), 4273--4286.


\bibitem{HE1989a}
I. Hauser and F. J. Ernst,  Initial value problem for colliding gravitational plane waves I, {\it J. Math. Phys.} {\bf 30} (1989), 872--887. 

\bibitem{HE1989b}
I. Hauser and F. J. Ernst, Initial value problem for colliding gravitational plane waves II, {\it J. Math. Phys.} {\bf 30} (1989), 2322--2336. 

\bibitem{HE1990}
I. Hauser and F. J. Ernst, Initial value problem for colliding gravitational plane waves III, {\it J. Math. Phys.} {\bf 31} (1990), 871--881. 

\bibitem{HE1991}
I. Hauser and F. J. Ernst, Initial value problem for colliding gravitational plane waves IV, {\it J. Math. Phys.} {\bf 32} (1991), 198--209.


\bibitem{KP1971}
K. A. Khan and R. Penrose,  Scattering of two impulsive gravitational plane waves, {\it Nature} {\bf 229} (1971), 185--186.

\bibitem{KMM2010} B. Konopelchenko, L. Mart\'{i}nez Alonso, E. Medina, Hodograph solutions of the dispersionless coupled KdV hierarchies, critical points and the Euler--Poisson--Darboux equation, {\it J. Phys. A: Math. Theor.} {\bf 43} (2010), 434020.

\bibitem{M1973}
W. Miller Jr., Symmetries of differential equations. The hypergeometric and Euler--Darboux equations, {\it SIAM J. Math. Anal.} {\bf 4} (1973), 314--328.

\bibitem{NH1977}
Y. Nutku and M. Halil, Colliding impulsive gravitational waves, {\it Phys. Rev. Lett.} {\bf 39} (1977), 1379--1382.

\bibitem{PM2017}
S. Palenta and R. Meinel, A continuous Riemann-Hilbert problem for colliding plane gravitational waves, {\it Class. Quantum Grav.} {\bf 34} (2017), 195011.

\bibitem{PHNS2017a} N. Popivanov, T. Hristov, A. Nikolov, M. Schneider, Singular Solutions to a $(3 + 1)$-D Protter--Morawetz Problem for Keldysh-Type Equations, {\it Advances in Mathematical Physics} {\bf 2017} (2017), Article ID 1571959.

\bibitem{PHNS2017b} N. Popivanov, T. Hristov, A. Nikolov, M. Schneider, On the existence and uniqueness of a generalized solution of the Protter problem for $(3+1)$-D Keldysh-type equations, {\it Boundary Value Problems} {\bf 26} (2017), 1--30.

\bibitem{S1972}
P. Szekeres,  Colliding plane gravitational waves, {\it J. Math. Phys.} {\bf 13} (1972), 286--294.


\end{thebibliography}

\end{document}